\documentclass[twocolumn,aps,pra,groupedaddress,superscriptaddress,nofootinbib,showpacs]{revtex4-1}

 \pdfoutput=1
  
\usepackage{amsmath}
\usepackage{amssymb}
\usepackage{amsthm}
\usepackage[hyperindex,breaklinks]{hyperref}
\usepackage{graphicx}
\usepackage{graphics}
\usepackage{color}

\newcommand{\ZZ}{\mathbb{Z}}

\newcommand{\AC}{\mathcal{A}}

\newcommand{\CC}{\mathcal{C}}

\newcommand{\FC}{\mathcal{F}}
\newcommand{\GC}{\mathcal{G}}
\newcommand{\HC}{\mathcal{H}}
\newcommand{\IC}{\mathcal{I}}
\newcommand{\JC}{\mathcal{J}}

\newcommand{\LC}{\mathcal{L}}

\newcommand{\NC}{\mathcal{N}}

\newcommand{\PC}{\mathcal{P}}
\newcommand{\QC}{\mathcal{Q}}

\newcommand{\SC}{\mathcal{S}}
\newcommand{\TC}{\mathcal{T}}

\newcommand{\ket}[1]{|#1\rangle}                  
\newcommand{\bra}[1]{\left\langle #1 \right|}     
\newcommand{\dyad}[2]{\ket{#1}\bra{#2}}           
\newcommand{\matl}[3]{\langle #1|#2|#3\rangle}    
\newcommand{\Tr}{{\rm Tr}}                        
\newcommand{\vect}[1]{\vec{#1}}         

\newcommand{\ii}{\mathrm{i}}					  

\def\dya#1{|#1\rangle \langle#1|}

\long\def\ca#1\cb{} 

\newtheorem{theorem}{Theorem}
\newtheorem{lemma}{Lemma}

\newcommand{\pt}{S}
\newcommand{\ptc}{\bar S}

\begin{document}

\title{Generalized Semi-Quantum Secret Sharing Schemes}

\author{Vlad Gheorghiu}
\email{vgheorgh@ucalgary.ca}
\affiliation{Institute for Quantum Information Science and}
\affiliation{Department of Mathematics and Statistics,\\University of Calgary, 2500 University Drive NW,\\Calgary, AB, T2N 1N4, Canada}


\begin{abstract}
We investigate quantum secret sharing schemes constructed from $[[n,k,\delta]]_D$ non-binary stabilizer quantum error correcting codes with carrier qudits of prime dimension $D$. We provide a systematic way of determining the access structure, which completely determines the forbidden and intermediate structures.
We then show that the information available to the intermediate structure can be fully described and quantified by what we call the \emph{information group}, a subgroup of the Pauli group of $k$ qudits, and employ this group structure to construct a method for hiding the information from the intermediate structure via twirling of the information group and sharing of classical bits between the dealer and the players. Our scheme allows the transformation of a ramp (intermediate) quantum secret sharing scheme into a semi-quantum perfect secret sharing scheme with the same access structure as the ramp one but without any intermediate subsets, and is optimal in the amount of classical bits the dealer has to distribute. 
\end{abstract}

\maketitle


\section{Introduction\label{sct1}} 
Classical secret sharing, introduced first by \cite{Blakely} and \cite{Shamir:1979}, is an important multipartite cryptographic protocol in which a dealer distributes a secret to a set of participants (players) in such a way that only certain subsets of players that form the \emph{access structure}  can collaboratively recover it. Quantum secret sharing \cite{PhysRevA.59.1829,PhysRevA.61.042311,PhysRevLett.83.648} is the natural extension of the classical protocol to the quantum domain: the secret is now a quantum state, the players comprise of quantum systems and quantum communication is allowed between the dealer and the players. The $(q,n)$ \emph{threshold} quantum secret sharing scheme is one of the most common protocols, in which the access structure comprises all subsets of $q$ or more out of $n$ players, and the forbidden structure consists of all subsets of less than $q$ players. Recently the threshold quantum schemes have been extended to \emph{intermediate}, or \emph{ramp} schemes \cite{PhysRevA.72.032318}, in which there are subsets of players that may recover partial information about the secret and which collectively form the \emph{intermediate structure}. Ramp schemes trade security for efficiency: they allow the sharing of quantum secrets of a dimension larger than the dimension of the players' shares, which, as we explain later, is impossible in threshold schemes.

An important desideratum in the theory of quantum (as well as classical) secret sharing is the construction of good protocols, and a vast amount of work is dedicated to this subject \cite{PhysRevA.71.044301,PhysRevA.78.042309,  PhysRevA.81.052333, PhysRevA.82.062315, quantph.1108.5541}. A promising approach is the using quantum error correcting codes for the construction of quantum secret sharing protocols: recovering a quantum secret is equivalent to the ability to detect and correct errors. One of the simplest examples of such duality between error correcting codes and secret sharing is the $[[5,1,3]]_2$ qubit code that induces a $(3,5)$ threshold quantum secret sharing scheme \cite{PhysRevLett.83.648}. 

Given some quantum error correcting code, a fundamental problem is to determine the induced access, forbidden and intermediate structures, and to quantify the information available in intermediate subsets. The vast majority of literature approaches this problem from a state point of view: consider an arbitrary quantum state $\ket{\psi}$ on the input state, then investigate the reduced density matrix of the encoded state down to some subset $S$ of the carriers. If the reduced density matrix is independent of the input state, then $S$ belongs to the forbidden structure, whereas if the density matrix is isometrically equivalent to $\dya{\psi}$ then $S$ belongs to the access structure. The intermediate structure consists of subsets that do not satisfy the previous two conditions.

Although the above approach works, it is in general tedious. Since the input state (secret) consists of $k$ qudits, the number of coefficients used to describe it scale exponentially with $k$, and even numerical methods become inefficient for quantum codes with large $k$. In our present work, we employ a completely different approach based on a channel point of view, and regard the error correcting code as an isometric encoding of $k$ qudits into $n$ carriers, as shown schematically in Fig.~\ref{fgr1}.
\begin{figure}
	\includegraphics[scale=0.6]{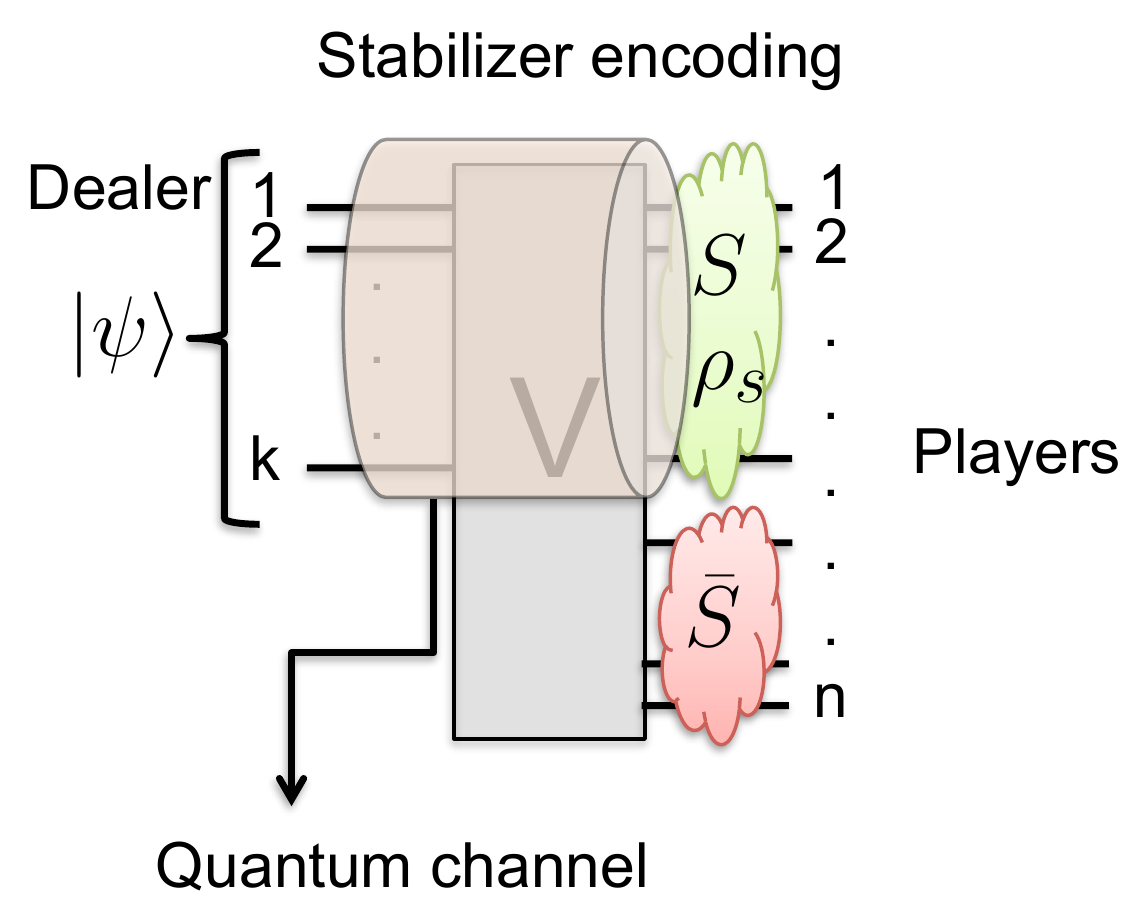}
	\caption{Quantum secret sharing scheme induced by a stabilizer code, where $V$ denotes the  encoding isometry.}
	\label{fgr1}
\end{figure}
If the channel from the input of the isometry to some subset $S$ of carriers is perfect (up to a unitary or isometry) then $S$ belongs to the access structure, and if it is totally noisy then $S$ belongs to the forbidden structure, with the intermediate case in between. For the class of qudit stabilizer codes, which include the vast majority of known error correcting codes, we showed \cite{PhysRevA.81.032326} that any such channel can be fully characterized by what we called an information group, a subgroup of the Pauli group of $k$ qudits. The symplectic structure of the information group fully characterizes the capacity of the channel and it can be shown that the latter can perfectly transmit an integer number of $r$ qudits plus an additional integer number $s$ of classical ``dits", with $r+s\leqslant k$. This allows us to precisely quantify the amount of accessible information by an intermediate set. Furthermore, determining whether the channel is perfect or not is a polynomial time (in $k$) decision problem, and this allows us to determine if $S$ belongs the access structure efficiently. Having the access structure determined, we show that the forbidden and intermediate structures are fully determined by the former.

We next show how to improve the security of generalized secret sharing schemes and transform them to perfect semi-quantum schemes, i.e. how to effectively ``remove" the intermediate structure while keeping the access structure the same. Our method is based on twirling what we call the \emph{intermediate information group}, a subgroup of the Pauli group of $k$ qudits associated with the intermediate structure.  The symplectic structure of this group provides a systematic method of erasing the intermediate structure by allowing the dealer to send classical information to the players using an appropriate classical secret sharing scheme. We show that our scheme is optimal in the amount of classical bits the dealer has to distribute. 

The remainder of this article is organized as follows. In Sec.~\ref{sct2} we define the generalized Pauli group and qudit stabilizer codes. We then show in Sec.~\ref{sct3} how any arbitrary stabilizer quantum error correcting code induces a generalized secret sharing scheme, and prove that the access structure of the latter fully determines the forbidden and intermediate structures. In Sec.~\ref{sct4} we introduce the subset information group, show that it fully characterizes the amount of information accessible by a subset, then present an algorithm for determining the access structure. The method of transforming an arbitrary generalized secret sharing scheme into a perfect scheme by allowing the sharing of extra classical bits between the dealer and the players is the subject of Sec.~\ref{sct5}. Finally, we present simple illustrative examples in Sec.~\ref{sct6} and conclusions and open questions in Sec.~\ref{sct7}.

\section{Preliminary remarks and definitions\label{sct2}}

\subsection{Generalized Pauli operators and graph codes}
We generalize Pauli operators to higher dimensional systems of prime dimension $D$ following \cite{patera:665,PhysRevA.64.012310,PhysRevA.65.052316}. 
The $X$ and $Z$ operators acting on a
single qudit are defined as
\begin{equation}
\label{eqn1}
Z=\sum_{j=0}^{D-1}\omega^j\dyad{j}{j},\quad X=\sum_{j=0}^{D-1}\dyad{j}{j+1},
\end{equation}
and satisfy
\begin{equation}
\label{eqn2}
X^D=Z^D=I,\quad XZ=\omega ZX,\quad \omega = \mathrm{e}^{2 \pi \ii /D},
\end{equation}
where \emph{the addition of integers is modulo $D$}, as will be 
assumed from now on. For a collection of $n$ qudits\footnote{or $k$, depending on the context; for the latter case one should replace $n$ by $k$ in all definitions of this subsection.}
we use subscripts to
identify the corresponding Pauli operators: thus $Z_i$ and $X_i$ operate on
the space of qudit $i$. The Hilbert space of a single qudit is denoted by $\HC$, and the Hilbert space of $n$ qudits by $\HC_n$, respectively. Operators of the form
\begin{equation}
\label{eqn3}
\omega^{\lambda}X^{\vect{x}}Z^{\vect{z}} :=
\omega^{\lambda}X_1^{x_1}Z_1^{z_1}\otimes X_2^{x_2}Z_2^{z_2}\otimes\cdots
\otimes X_n^{x_n}Z_n^{z_n}
\end{equation} 
will be referred to as \emph{Pauli products}, where $\lambda$ is an integer
in $\ZZ_D$ and $\vect{x}$ and $\vect{z}$ are $n$-tuples in $\ZZ_D^n$, the
additive group of $n$-tuple integers mod $D$. For a fixed $n$ the collection
of all possible Pauli products \eqref{eqn3} form a group under operator
multiplication, the \emph{Pauli group} $\PC_n$. If $p$ is a Pauli product,
then $p^D=I$ is the identity operator on $\HC_n$, and hence the order of any
element of $\PC_n$ is $D$. While
$\PC_n$ is not Abelian, it has the property that two elements \emph{commute up
  to a phase}
\begin{equation}\label{eqn4}
p_1p_2 = \omega^{\lambda_{12}} p_2p_1,
\end{equation}
with $\lambda_{12}$ an
integer in $\ZZ_D$ that depends on $p_1$ and $p_2$.


The collection of Pauli products with $\lambda=0$, i.e. a pre-factor of $1$, is
denoted by $\QC_n$. The elements $X^{\vect{x}}Z^{\vect{z}}$ of $\QC_n$ form an orthonormal basis of
$\LC(\HC_n)$, the Hilbert space of linear operators on
$\HC_n$, with respect to the inner product
\begin{align}
\label{eqn5}
\frac{1}{D^n}&\Tr[({X^{\vect{x}_1}Z^{\vect{z}_1}})^\dagger X^{\vect{x}_2}Z^{\vect{z}_2}]=\delta_{\vect{x}_1,\vect{x}_2}\delta_{\vect{z}_1,\vect{z}_2},\notag\\
&\forall X^{\vect{x}_1}Z^{\vect{z}_1}, X^{\vect{x}_2}Z^{\vect{z}_2}\in \QC_n.
\end{align}
Note that $\QC_n$ is a \emph{projective group} or group up
to phases. There is a bijective map between $\QC_n$ and the quotient group
$\PC_n/\{\omega^{\lambda}{I}\}$ for $\lambda\in\ZZ_D$ where
$\{\omega^{\lambda}{I}\}$, the center of $\PC_n$, consists of phases
multiplying the identity operator on $n$ qudits. The projective group $\QC_n$ is also isomorphic to the additive group $\ZZ_D^{2n}$ of $2n$-tuple integers under addition mod $D$.

\subsection{Qudit stabilizer codes}
Relative to the Pauli group $\PC_n$ of $n$ \emph{carrier} qudits we define a \emph{stabilizer} code $\HC_C$
to be a $K\geq 1$-dimensional subspace of the carriers' Hilbert space $\HC_n$, $\HC_C\subset\HC_n$, 
satisfying three conditions:

\begin{description}

\item[C1] There is a subgroup $\SC$ of $\PC_n$ such that
for \emph{every} $s$ in $\SC$ and \emph{every} $\ket{\psi}$ in
$\CC$
\begin{equation}
\label{eqn6}
 s \ket{\psi} = \ket{\psi}
\end{equation}

\item[C2] The subgroup $\SC$ is maximal in the sense that every
$s$ in $\PC_n$
for which \eqref{eqn6} is satisfied for all $\ket{\psi}\in\HC_C$
belongs to
  $\SC$.

\item[C3] The coding space $\HC_C$ is maximal in the sense that
any ket
$\ket{\psi}$ that satisfies \eqref{eqn6} for every $s\in\SC$
lies in
  $\HC_C$.
\end{description}

If these conditions are fulfilled we call $\SC$ the
\emph{stabilizer} of the
code $\HC_C$. That it is Abelian follows from the commutation
relation \eqref{eqn4}, since for $K>0$
there is some nonzero $\ket{\psi}$ satisfying \eqref{eqn6}.  

Note that one
can always find a subgroup $\SC$ of $\PC_n$ satisfying C1 and C2
for any
subspace $\HC_C$ of the Hilbert space, but it might consist of
nothing but the
identity. Thus it is condition C3 that distinguishes stabilizer
codes from
nonadditive codes. A stabilizer code is uniquely determined by
$\SC$ as
well as by $\HC_C$, since $\SC$ determines $\HC_C$ through C3, so in
a sense the code and its stabilizer are dual
to each other.

A qudit stabilizer code is usually denoted by $[[n,k,\delta]]_D$ where $n$ represents the number of carrier qudits (each of dimension $D$, assumed in this article a prime number) and $k$ specifies the number of \emph{input} (or encoded) qudits, assumed to have the same dimensionality $D$ as the carrier qudits\footnote{See \cite{quantph.1101.1519} for a general treatment of stabilizer codes with qudits of composite dimensionality, where some differences arise. For example, the input qudits do not have to have the same dimensionality $D$ as the carriers, but can be of any dimension $d$ that divides $D$.}. Here $\delta$ is the \emph{distance} of the code \cite{NielsenChuang:QuantumComputation}, a parameter that essentially specifies how ``good" the code is: best codes have as large as possible distance with as few as possible carriers. The stabilizer code can then be seen as arising from the isometric encoding of the input space $\HC_k$ into the $K=D^k$ dimensional subspace $\HC_C$ of $\HC_n$ by the isometry
\begin{align}\label{eqn7}
V:\HC_k\longrightarrow \HC_C\subset\HC_n,\quad
V=\sum_{j=0}^{K-1}\dyad{c_j}{j}.
\end{align}
Here $\{\ket{j}\}$ is an orthonormal basis of the input space $\HC_k$ and the coding space is specified by $\HC_C=\mathrm{Span}\{\ket{c_j}\}$, where the \emph{codewords} $\ket{c_j}$'s are all orthogonal. Note that $V$ cannot be any isometric encoding of $\HC_k$ into $\HC_C$, but one compatible with the stabilizer requirements C1--C3 above.

\section{Generalized secret sharing schemes}\label{sct3}
\subsection{Access, forbidden and intermediate structures}
We now show that any quantum error correcting code can be turned into a generalized secret sharing scheme as follows. 
Let $\ket{\psi}\in\HC_k$ be an arbitrary $k$ qudit quantum state, (secret, unknown by the $n$ output qudits), which is then ``distributed" to the $n$ carrier qudits (players) via the corresponding stabilizer encoding $V$, so the $n$ players end up sharing the encoded state $V\ket{\psi}\in\HC_C\subset\HC_n$. Let $N=\{1,2,\ldots, n\}$ denote the set of output qudits, and let $\NC$ be the collection of all subsets of $N$, i.e. the power set of $N$. We define the following structures:
\begin{itemize}
\item $\AC$ -- the \emph{access structure} (or the \emph{authorized structure}): $\AC\subset \NC$ such that any subset of qudits $S\in\AC$ can fully recover $\ket{\psi}$, i.e. the quantum channel from the input of the isometry $V$ to any subset in $S\in\AC$ is perfect.
\item $\FC$ -- the \emph{forbidden structure} (or the \emph{un-authorized structure}): $\FC\subset \NC$ such that no subset of qudits $S\in\FC$ can recover anything about $\ket{\psi}$, i.e. the quantum channel from the input of the isometry $V$ to any subset in $S\in\FC$ is completely noisy.
\item $\IC$ -- the \emph{intermediate structure} (or the \emph{ramp structure}): $\IC\subset \NC$ such that any subset of qudits $S\in\IC$ can recover some partial information about $\ket{\psi}$, i.e. the quantum channel from the input of the isometry $V$ to any subset in $S\in\AC$ is noisy (not perfect nor completely noisy).
\end{itemize}
In conclusion, the isometry $V$ completely determines the triplet $(\AC,\FC,\IC)$, and we call the latter a \emph{generalized secret sharing scheme}. Whenever $\IC=\{\emptyset\}$ we call the scheme \emph{perfect}. Our definition generalizes the two most common secret sharing schemes in the literature: 
\begin{enumerate}
\item the \emph{threshold $(q,n)$ secret sharing scheme}, in which any subset of $q$ or more players can fully recover the quantum secret (are authorized), whereas any subset of less than $q$ players cannot recover any information whatsoever about the secret (are forbidden). Formally,
\begin{align}\label{eqn8}
\AC&=\left\{ S\in \NC : |S|\geqslant q\right\},\notag\\
\FC&=\left\{ S\in \NC : |S|< q\right\},\notag\\
\IC&=\{\emptyset\},
\end{align}
where $|S|$ denotes the size of the set $S$, i.e. the number of players in $S$. The threshold schemes are a strict subset of the perfect schemes, since the latter allow for access structures with subsets of different sizes;
\item the \emph{ramp $(q,L,n)$ secret sharing scheme}, in which \emph{any} subset of $q$ or more players is authorized, \emph{any} subset of $q-L$ or fewer is forbidden, and those with $q-u (0<u<L)$ are not all authorized nor all forbidden. In this notation a threshold secret sharing scheme has \mbox{$L=1$}.
In our notation, a ramp $(q,L,n)$ scheme must have
\begin{align}\label{eqn9}
\AC&\supseteq\left\{ S\in \NC : |S|\geqslant q\right\},\notag\\
\FC&\supseteq\left\{ S\in \NC : |S|\leqslant q-L\right\},\notag\\
\IC&\subseteq\{ S\in \NC : q-L<|S|<q \}.
\end{align}
\end{enumerate} 
Note that any $[[n,k,\delta]]_D$ stabilizer code can be turned into a $(q,L,n)$ ramp secret sharing scheme, with $q=n-\delta+1$ and $L=n-2\delta+2$, since any subset of more than \mbox{$n-\delta$} players  has full information about the secret and can fully recover the secret by a suitable decoding procedure, hence it is an authorized set and belongs to the access structure, whereas any subset of less than $\delta$ players has no information whatsoever about the secret, hence it is a forbidden set and belongs to the adversary structure; see Sec.~III.A of \cite{PhysRevA.56.33} for a simple no-cloning based argument. One therefore has a ramp $(n-\delta+1,n-2\delta+2,n)$ quantum secret sharing scheme.  

We now illustrate the concepts of this section by simple examples. 
First, consider the $[[5,1,3]]_2$ code \cite{PhysRevLett.77.198,PhysRevA.54.3824}. It can be shown that all subsets of size 3 or more can fully recover whatever information was encoded, whereas any subset of size 2 or less cannot recover anything. Therefore this code can be turned into a $(q=3,n=5)$ threshold secret sharing scheme.

Next, consider the $[[7,1,3]]_2$ additive graph code of \cite{PhysRevA.81.032326}, which is locally unitarily equivalent to the Steane code \cite{PhysRevLett.77.793}. It then follows that all subsets of size 5 or more can recover what was encoded, whereas any subset of size 2 or less cannot recover anything. However, there is more to say about this code, and one can show that the subsets of size 3 or 4 can either fully recover the secret or cannot recover anything, hence $\IC=\{\emptyset\}$, so the scheme is perfect (although not threshold). See the discussion on pg. 10 of \cite{PhysRevA.81.032326} for a detailed discussion and for a full list of subsets comprising $\AC$ and $\FC$.

Finally, consider the $[[4,2,2]]_2$ code \cite{PhysRevA.54.R1745} that can correct one erasure error, i.e. can fully correct one qubit error provided one knows what the corrupted qubit is. It can be shown \cite{PhysRevA.81.032326} that all subsets of size 3 or 4 can recover all encoded information, whereas any subset of size 1 cannot recover anything. All subsets of size 2 can only recover partial information (are not able to fully reconstruct what was encoded). We can therefore turn this code into a $(q=3,L=2,n=4)$ ramp secret sharing scheme. Using our formalism, $\AC=\{S\in \NC : |S|\geqslant 3\}$, $\FC=\{S\in \NC : |S| \leqslant 1\}$ and $\IC=\{S\in \NC : |S|=2\}$.

\subsection{Relations between $\AC$, $\FC$ and $\IC$}
The following question arrises naturally: given an arbitrary stabilizer code, how can one determine the induced triplet $(\AC,\FC,\IC)$? This question is of crucial importance in the theory of quantum secret sharing, and, before providing a full answer to this question, we first explain why the access structure $\AC$ completely determines the forbidden structure $\FC$ (and viceversa), and together they determine the intermediate structure $\IC$, so it is enough to know only $\AC$ (or $\FC$) to determine the full triplet $(\AC,\FC,\IC)$. The fact that $\AC$ is dual to $\FC$ was already known \cite{PhysRevLett.83.648,PhysRevA.61.042311}, but we restate the result for the sake of completeness. 
\begin{theorem}\label{thm1}
Let $(\AC,\FC,\IC)$ be a generalized quantum secret sharing scheme induced by some quantum error correcting code. Then $\AC$ and $\FC$ are dual to each other, and completely determine $\IC$, in the sense
\begin{align}
\AC&=\left\{ S\in \NC : \bar S\in\FC\right\}\label{eqn10},\\
\FC&=\left\{ S\in \NC : \bar S\in\AC\right\}\label{eqn11}, \\
\IC&=\NC\setminus\left\{\AC\bigcup\FC \right\},\label{eqn12}
\end{align}
where $\bar S=N\setminus S$ denotes the complement of $S$. 
\end{theorem}

\begin{proof}
The argument is based on two facts: i) a perfect quantum channel cannot ``leak" information, since otherwise the no-cloning theorem is violated; and ii), an isometry ``conserves" information: if it is absent from some part of its output it has to be present in the complement. 

Consider first a subset $S\in\AC$. Then the players in $S$ can recover full information about what was encoded, and, by the ``No Splitting" Theorem of \cite{PhysRevA.76.062320}, this implies that the complement $\bar S$ cannot contain any information whatsoever about the input, hence must belong to the forbidden structure $\FC$. Intuitively, if $S$ belong to the access structure, then the quantum channel from the input of the underlying isometry $V$ to $S$ (obtained by partially-tracing down $\bar S$) must be perfect, so the complementary channel to $\bar S$ must be completely noisy (otherwise the no-cloning theorem will be violated), which is the same as saying that $\bar S\in \FC$.

On the other hand, let's now consider a subset $S\in\FC$. Then, absolutely no information about what was encoded can be recovered from $S$, and, by the ``Somewhere Theorem" of \cite{PhysRevA.76.062320}, all information about the input must be present in $S$, hence $S\in\AC$. Intuitively, this is the same as saying that an isometry ``conserves" information: if it is absent from a subset it \emph{must} be present in its complement. This proves the duality \eqref{eqn10}--\eqref{eqn11} between $\AC$ and $\FC$.

Finally, \eqref{eqn12} follows at once by construction.
\end{proof}

We therefore conclude this section by restating that it is enough to determine $\AC$ (or $\FC$) in order to fully determine $(\AC,\FC,\IC)$. In the next section we provide a systematic way of determining $\AC$.

\section{Determining $(\AC,\FC,\IC)$}\label{sct4}
We now review some essential results about information location in subsets of players of a quantum secret sharing scheme induced by an $[[n,k,\delta]]_D$ stabilizer code. The interested reader can consult our previous work \cite{PhysRevA.81.032326} for detailed proofs of the claims of this section\footnote{All results in  were proven for additive graph codes (a subset of stabilizer codes, see e.g. \cite{PhysRevA.78.042303} for a comprehensive introduction), but we noted that all our results are automatically valid for prime dimensional stabilizer codes, since the latter are locally unitary equivalent to the former, as proved by Schlingemann in \cite{quantph.0111080}.}.

\subsection{The subset information group and the access structure}
Let us consider a subset $\pt\in \NC$ of players, and let $\ptc$ denote its complement. The relevant question for quantum secret sharing is how much information can $\pt$ recover about a previously encoded secret $\ket{\psi}\in\HC_k$? Whatever information $\pt$ can extract about the secret is fully determined by the reduced density matrix
\begin{equation}\label{eqn13}
\rho_\pt:=\Tr_{\ptc}[V\dya{\psi}V^\dagger].
\end{equation}
Since the collection $\QC_k$ of Pauli operators on the input space $\HC_k$ forms an operator basis of the dealer's operator Hilbert space $\LC(\HC_k)$, one can expand the secret as
\begin{equation}\label{eqn14}
\dya{\psi}=\frac{1}{D^k}\sum_{\vect{x},\vect{z}\in\ZZ_D^k}c(\vect{x},\vect{z})X^{\vect{x}}Z^{\vect{z}},
\end{equation}
where 
\begin{equation}\label{eqn15}
c(\vect{x},\vect{z})=\Tr\left[({X^{\vect{x}}Z^{\vect{z}}})^\dagger\dya{\psi}\right]=\matl{\psi}{({X^{\vect{x}}Z^{\vect{z}}})^\dagger}{\psi}
\end{equation}
are the Fourier coefficients of the expansion.

The state of $\pt$ is then
\begin{equation}\label{eqn16}
\rho_\pt=\sum_{\vect{x},\vect{z}\in\ZZ_D^k}c(\vect{x},\vect{z}) \Tr_{\ptc}\left[VX^{\vect{x}}Z^{\vect{z}}V^\dagger\right].
\end{equation}

We shown in \cite{PhysRevA.81.032326} that the collection of operators $X^{\vect{x}}Z^{\vect{z}}$ on the dealer's space $\HC_k$ for which \mbox{$\Tr_{\ptc}\left[VX^{\vect{x}}Z^{\vect{z}}V^\dagger\right]\neq0$} forms a group $\GC(\pt)$, called the \emph{subset information group}. We have also proved that the subset information group fully characterizes what kind of correlations are present between the dealer and the subset $\pt$ of the players, and provided an efficient linear algebra based algorithm for finding it. More specifically, if $X^{\vect{x}}Z^{\vect{z}}\in\GC(\pt)$, i.e.  $\Tr_{\ptc}\left[VX^{\vect{x}}Z^{\vect{z}}V^\dagger\right]\neq0$, then any two eigenvectors $\ket{\phi_1}$ and $\ket{\phi_2}$ of $X^{\vect{x}}Z^{\vect{z}}$ remain fully distinguishable on the subset $\pt$ after the encoding by $V$, i.e. have orthogonal support so their Hilbert-Schmidt inner product is zero
\begin{equation}\label{eqn17}
\Tr\left[ \left(\Tr_{\ptc}[V\dya{\phi_1}V^\dagger]\right)^\dagger \Tr_{\ptc}[V\dya{\phi_2}V^\dagger] \right]=0.
\end{equation}
In other words, if $X^{\vect{x}}Z^{\vect{z}}\in\GC(\pt)$, the correlations between the players in $\pt$ and the dealer are perfect in the eigenbasis of $X^{\vect{x}}Z^{\vect{z}}$, that is, if the dealer chooses the secret to be one of the the eigenvectors of $X^{\vect{x}}Z^{\vect{z}}$, say $\ket{\phi_j}$, then the players in $\pt$ can fully recover the $j$ by performing an appropriate positive operator-valued measure (POVM). We say that the $X^{\vect{x}}Z^{\vect{z}}$ type of information\cite{PhysRevA.76.062320} is \emph{perfectly present} in $\pt$. We also proved that the $C^{*}$-algebra generated by the elements of $\GC(\pt)$ is fully correctable\cite{PhysRevLett.98.100502,PhysRevA.76.042303,quantph.0907.4207} on $\pt$, that is, any 2 orthogonal states in the algebra remain orthogonal after encoding and tracing down to $\pt$.
 
The subset $\pt$ contains no information whatsoever about the secret $\dya{\psi}$ if and only if the subset information group is proportional to identity on $\HC_k$, $\GC(\pt)\propto I$, or, equivalently, all encoded operators $VX^{\vect{x}}Z^{\vect{z}}V^\dagger$ trace to zero down to $\pt$ with the exception of $VX^{\vect{0}}Z^{\vect{0}}V^\dagger=VV^\dagger$. That is, no matter what measurement strategy the players in $\pt$ adopt, they cannot recover any information about the secret $\dya{\psi}$, or, equivalently, the reduced density matrix $\rho_\pt$ is independent of $\dya{\psi}$,
\begin{equation}\label{eqn18}
\rho_\pt=\frac{1}{D^k}\Tr_{\ptc}\left[ VV^\dagger \right].
\end{equation}
The subset $\pt$ contains all information about the secret if and only if the subset information group is the whole Pauli group $\PC_k$, $\GC(\pt)=\PC_k$, hence we have the following Theorem.
\begin{theorem}\label{thm2}
Let $(\AC,\FC,\IC)$ be a generalized quantum secret sharing scheme induced by an $[[n,k,\delta]]_D$ quantum error correcting code. Then the access structure $\AC$ is given by
\begin{equation}\label{eqn19}
\AC=\left\{ S \in \NC : \GC(\pt)=\PC_k\right\}.
\end{equation}
\end{theorem}
\begin{proof}
See Theorem~4 (iii) of \cite{PhysRevA.81.032326} for a rigorous proof.
\end{proof}
To determine whether $\GC(\pt)=\PC_k$ for some subset $\pt$ reduces to checking whether the partial trace down to $\ptc$ of the $2k$ encoded generators of $\PC_k$ is not zero. This is because if some encoded generator traces down to zero on $\pt$, then, by the group property, $\GC(S)$ must be a strict subset of $\PC_k$ (removing an independent generator makes the generated group strictly smaller). The question ``Is $\GC(\pt)=\PC_k$?" is a decision problem, and its yes/no answer can be provided via solving a system of linear equations over $\ZZ_D$,  see Appendix~C of \cite{PhysRevA.81.032326} for a detailed efficient algorithm (with polynomial running time in $n$ and $k$).

Note that one can also use the Choi-Jamio\l kowski isomorphism in determining if a subset $\pt$ belongs to the access structure as follows. Consider a maximally entangled state $\ket{\Psi^{+}}$ between the input of the isometry and some reference system $R$. Let
\begin{equation}\label{eqn20}
\Omega=(I_R\otimes V)\ket{\Psi^{+}}
\end{equation}
and let
\begin{equation}\label{eqn21}
\rho_{R\pt}=\Tr_{\ptc}\dya{\Omega}.
\end{equation}
Then $\pt\in\AC$ if and only if $\rho_{R\pt}$ is a pure maximally entangled state, since the latter implies that the channel from the input of $V$ to $\pt$ is perfect. 
However, our approach is more powerful since it characterizes the information present in intermediate subsets, as described in the next subsection.

\subsection{The structure of the subset information group and the information available in an intermediate subset}
Let $g_1,g_2,\ldots,g_m$ be a minimal generating set\footnote{That is, removing any generator results in generating a strictly smaller group.} of $\GC(\pt)$,
\begin{equation}\label{eqn22}
\GC(\pt)=\langle g_1,g_2,\ldots,g_m\rangle.
\end{equation}
Since $\GC(\pt)$ is a subgroup of the Pauli group of $k$ qudits each of prime dimension, it is Clifford equivalent to a simpler group $\GC_0(\pt)$, the \emph{canonical subset information group}, generated only by ``local" $X$ and $Z$ operators,
\begin{align}
\GC_0(\pt)&\equiv W\GC(\pt)W^\dagger\notag\\
&=\langle X_1, Z_1,\ldots, X_r,Z_r, Z_{r+1},\ldots, Z_{r+s}\rangle\label{eqn23}\\
&=\GC_0^{sym}(\pt)\bigcup \GC_0^{iso}(\pt), \text{ with } 2r+s\leqslant 2k\label{eqn24},
\end{align}
where $W$ is a Clifford operator\footnote{A Clifford operator in a unitary operator that maps Pauli operators to Pauli operators through conjugation, that is, leaves the Pauli group invariant under conjugation.} that depends on the subset $\pt$, but for simplicity of notation we remove this dependence. Here
\begin{equation}\label{eqn25}
\GC_0^{sym}(\pt)=\langle X_1,Z_1,\ldots, X_r,Z_r\rangle
\end{equation}
is the \emph{symplectic} subgroup of $\GC_0(\pt)$ and
\begin{equation}\label{eqn26}
\GC_0^{iso}(\pt)=\langle Z_{r+1},\ldots, Z_{r+s}\rangle
\end{equation}
is the \emph{isotropic} subgroup of $\GC_0(\pt)$.

The fact that $\GC(\pt)$ is isomorphic to $\GC_0(\pt)$ in \eqref{eqn23} is a direct consequence of a more general result regarding the structure of bilinear symplectic forms, see Theorem 1.1 of \cite{daSilva:LecNotesSymplecticGeometry}. Sec. IV.B of \cite{quantph.1105.5872} provides an explicit algorithm for constructing the Clifford operator $W$ as a product of elementary qudit Clifford gates; the qudit algorithm is just a straightforward generalization of the qubit one presented in Sec. 1 of \cite{quantph.0608027}. 

The form of $\GC_0(\pt)$ combined with the remarks of the previous subsection that the $C^{*}$-algebra generated by $\GC_0(\pt)$ is fully correctable allows us to say that the quantum channel from the input of the isometry to the subset $\pt$, obtained by partially-tracing over $\ptc$, is a perfect $r$-qudit channel (corresponds to the symplectic subgroup $\GC_0^{sym}(\pt)$) tensored with a perfect $s$ dit classical channel (corresponds to the isotropic subgroup $\GC_0^{iso}(\pt)$ and its quantum capacity is zero, since its correctable algebra is commutative). In other words, the channel can perfectly transmit $r$ qudits plus extra $s$ classical dits, which is equivalent, in the context of secret sharing, to the fact that the players in $\pt$ can fully recover $r$ qudits of the secret together with at most $s$ ``extra" classical dits of information by performing an appropriate decoding procedure.

\section{Concealing the intermediate structure via twirling}\label{sct5}
Consider now a generalized secret sharing scheme in which there are no intermediate structures, i.e. $\IC=\{\emptyset\}$. In this case it can be shown \cite{PhysRevA.61.042311} that the dimension of the quantum secret cannot exceed the dimension of each individual player's quantum system, or, equivalently, that $[[n,k,\delta]]_D$ codes induce generalized secret sharing schemes that must have $\IC\neq\{\emptyset\}$ whenever $k>1$. The argument is based on the fact that there exist forbidden subsets that can be made authorized by the addition of only one additional player, hence complete information about the secret is transferred via a single player's quantum system, from which the bound follows. Therefore, threshold quantum secret sharing schemes are extremely inefficient in terms of the required quantum communication. For example, if the dealer wants to share a $1,000$ qubit secret to $1,000$ players, then each player has to posses at least a $1,000$ qubit quantum system, for a total of $1,000\times 1,000=1,000,000$ carrier qudits! 

However, ramp schemes do not have this strong limitation: security is traded for efficiency, so that players belonging to the intermediate structure can extract some partial information about the secret, with the benefit that the encoded quantum secret can have larger dimension than the players' individual share size. 

Are there ways to improve the security of such intermediate secret sharing schemes, for example, by reducing the amount of information the intermediate structure can extract about the quantum secret? 
As recently shown in \cite{quantph.1108.5541} such methods exist and are based on combining the ramp quantum secret sharing scheme with a classical  secret sharing scheme. In the simplest scenarion, the dealer prepares a $k$-qudit secret, then for every input qudit $i$ chooses with equal probability $1/D$ two integers $m_i,n_i\in\ZZ_D$, then applies the operator $X_i^{m_i}Z_i^{n_i}$; the dealer effectively encrypts the quantum secret using a $2k$ classical key, then distributes the ``scrambled" secret to the players using the stabilizer encoding. From the players point of view, who do not know the integers $m_i$ and $n_i$, this is equivalent to the application by the dealer of a completely depolarizing, or ``twirling", channel to each input qudit. After this the dealer distributes the $2k$ dit classical key to the $n$ players using a classical threshold secret sharing scheme $(q,n)$, with $q$ appropriately chosen, so that any subset of $q$ or more players can recover the classical key, which allows them to ``undo" the effect of the depolarizing channel and recover the whole quantum secret, whereas any subset of less than $q$ players has no information whatsoever about the classical key and their shared quantum state is independent of the secret. In this way, a generalized secret sharing scheme induced by an $[[n,k,\delta]]_D$ stabilizer code is transformed to a threshold secret sharing scheme $(q,n)$. This method is very similar to teleportation, in which Bob cannot recover Alice's state without knowing the results of Alice's measurements (that play the role of the twirling channel).

This motivates the following question: given the intermediate structure $\IC$, what is the most efficient way of ``erasing" the information from it, i.e. what is the minimum length of the classical random encryption key the dealer must use and how can this be done systematically, for arbitrary generalized secret sharing schemes (induced by stabilizer quantum error correcting codes)? We present below such a systematic method and show that the dealer can in general use classical encryption keys of smaller  length $l$, with $k\leqslant l\leqslant 2k$, and that the length of the encryption key depends solely on the underlying stabilizer code. We will further show that our scheme is optimal in the length of the encryption key, that is, one cannot use shorter keys. This minimizes the amount of classical communication between the dealer and the players.

\subsection{Hiding information from a subset}
The whole idea behind our scheme is to employ the structure of the subset information group $\GC_0(S)$. Consider a generator of $\GC_0(\pt)$. Without loss of generality, we choose $X_1$. Then, as mentioned before, the players in $\pt$ can recover the $X_1$-type of information about the secret by an appropriate POVM. Suppose now that before the encoding the dealer applies with probability $p_j=1/D$ the operator $Z_1^{j}$, where $0\leq j<D$. One can easily check that
\begin{equation}\label{eqn27}
\frac{1}{D}\sum_{j=0}^{D-1}Z_1^j X_1 {Z_1^{j}}^\dagger=\frac{1}{D} \sum_{j=0}^{D-1}\omega^j X_1=0.
\end{equation}
Then, the $X_1$-type of information is hidden from the subset $\pt$, since, effectively, the players in $\pt$ do not know which $j$ the dealer chose, and, by \eqref{eqn27}, the terms of the form $\Tr_{\ptc}\left[\tilde VX_1\tilde V^\dagger\right]$ in the expansion \eqref{eqn16} of $\rho_{\pt}$ become now
\begin{equation}\label{eqn28}
\Tr_{\ptc}\left[\tilde V\left(\frac{1}{D}\sum_{j=0}^{D-1}Z_1^j X_1 {Z_1^{j}}^\dagger
\right)\tilde V^\dagger\right]=0,
\end{equation}
where $\tilde V = VW$ (we remind the reader that $\GC_0(\pt)$ was obtained from $\GC(\pt)$ via a conjugation by the Clifford operator $W$ and this is why $V$ is replaced by $\tilde V$ in \eqref{eqn28}).
Also, any other operator in the information group that contains $X_1$ gets mapped to zero (a direct consequence of \eqref{eqn27}), so at the end all non-zero terms in the expansion of $\rho_{\pt}$ that contain the $X_1$ generator become zero.

It is now clear how the dealer can hide all information about the secret from the players in $\pt$: twirls each generator of the canonical subset information group by a corresponding non-commuting operator (either $X$ or $Z$). Since any operator in the information group is a product of the generators, it follows at once that it will get mapped to zero by the twirling procedure, with the exception of the identity. Therefore, if before twirling the quantum secret was represented in $\pt$ by a state of the form \eqref{eqn16}, after the twirling procedure the state down to $\pt$ has the form \eqref{eqn18}, i.e. the players in $\pt$ lack any information whatsoever about the secret.

Note that if the players know which operators the dealer applied to perform the twirling, they can recover the same information about the secret as before the twirling, since the  twirling unitary on the dealer's space is effectively just a change of basis now known by the players.

\subsection{Perfect semi-quantum secret sharing schemes}
Consider now the collection of intermediate subsets $\IC$. To hide the information from $\IC$ we can define the ``intermediate information group" or the ``ramp information group" as a union of all intermediate subset information groups
\begin{equation}\label{eqn29}
\GC(\IC):=\bigcup_{\pt\in \IC}\GC(\pt),
\end{equation}
which again is a subgroup of the Pauli group $\PC_k$. Next employ the same arguments as in the previous subsection, but now with $\GC(\IC)$ replacing $\GC(\pt)$ (and $\GC_0(\IC)$ denoting the canonical intermediate information group obtained from $\GC(\pt)$ through an appropriate Clifford conjugation). The dealer distributes the randomly generated classical key to the players using a perfect classical secret sharing scheme $(\AC',\FC',\IC')$, having $\AC'=\AC$, $\FC'=\FC\bigcup\IC$ and $\IC'=\{\emptyset\}$, so that the players in $\IC$ have no information whatsoever about the key but players in $\AC$ can fully recover the key. In this way, the information is concealed from $\IC$. Note that classical secret sharing schemes with arbitrary access structure exist \cite{BenalohLeichter} as long as the access structure is monotone --i.e., if a set $S$ can recover the secret, so can all sets containing $S$. In conclusion, using our scheme one can transform an arbitrary $(\AC,\FC,\IC)$ ramp scheme into a perfect scheme $(\AC,\FC\bigcup\IC,\{\emptyset\})$ with the same access structure but without any intermediate subsets! See Fig.~\ref{fgr2} for a graphical description of our protocol.

In particular, if we define $q$ to be the minimum integer so that all subsets of players of size $q$ or more belong to $\AC$ (and therefore the subsets in $\IC$ must be of smaller size), then we can employ a classical $(q,n)$ threshold secret sharing scheme to distribute the classical key to the $n$ players so that all subsets in $\AC$ of size $q$ or more can fully recover the key, hence the full quantum secret, whereas the players in $\IC$ have no information whatsoever about the classical key and cannot recover anything about the original quantum secret. In this way, a generalized $(\AC,\FC,\IC)$ quantum secret sharing scheme induced by an $[[n,k,\delta]]_D$ stabilizer quantum error correcting code can be turned into a threshold $(q=n-\delta+1,n)$ semi-quantum secret sharing scheme by allowing the sharing of $2r+s\leqslant 2k$ classical bits between the dealer and the players. 
\begin{figure}
	\includegraphics[scale=0.63]{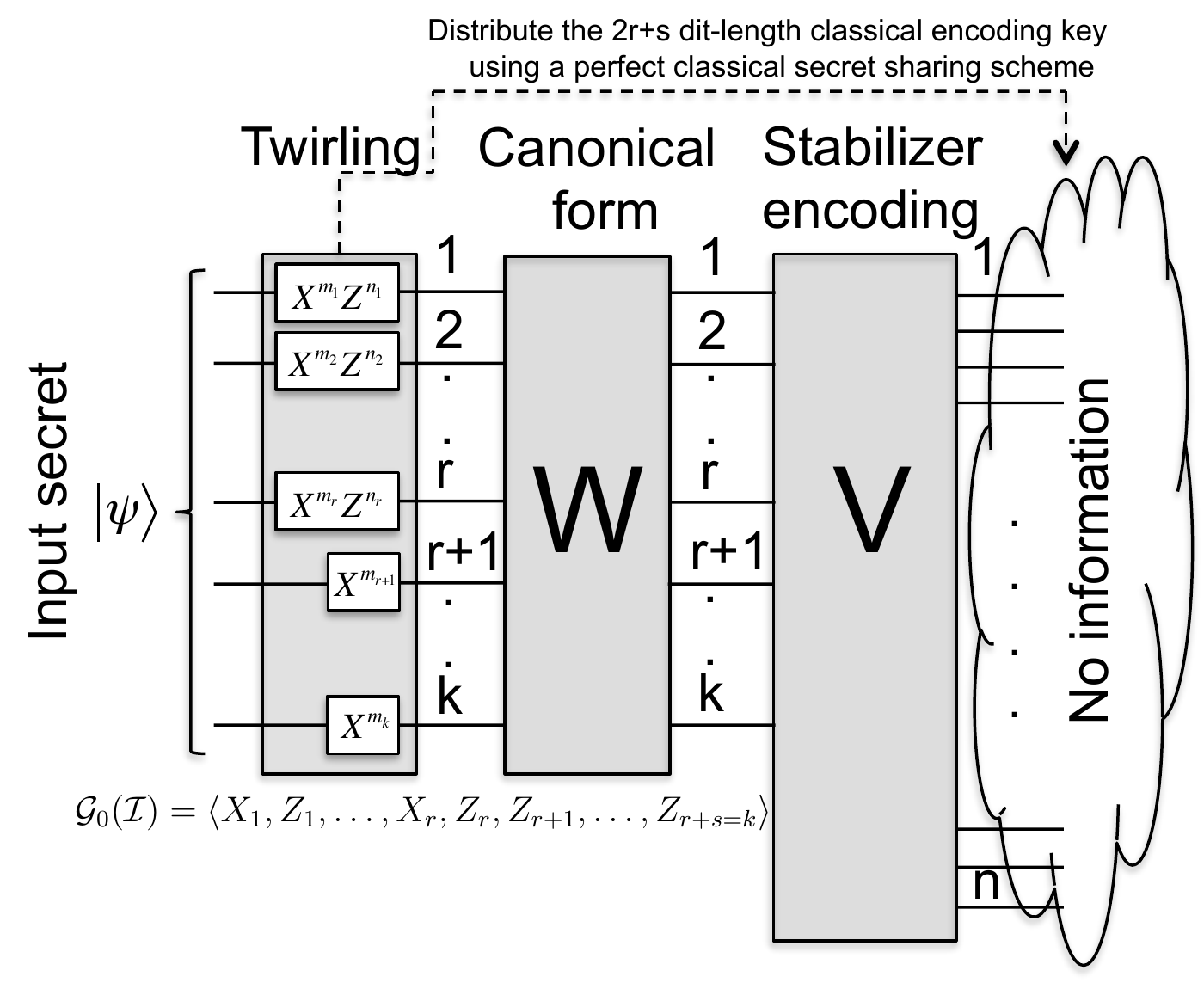}
	\caption{Turning a ramp quantum secret sharing scheme into a perfect semi-quantum secret sharing scheme. Note that $r+s=k$, as shown by Lemma~\ref{lma1}.}
	\label{fgr2}
\end{figure}

The twirling followed by the encoding of the classical key into a perfect classical secret sharing scheme can be seen as effectively ``cutting" the ramp structure and transforming it to an un-authorized structure. 
Mathematically, the intermediate information group $\GC_0(\IC)$ is being ``twirled" to the identity operator, i.e.
\begin{equation}\label{eqn30}
\sum_{U} UgU^\dagger=0,\quad \forall g\in\GC_0(\IC), g\neq I,
\end{equation}
where the sum is taken over all unitary operators $U$ of the form $X_i^{m_i}Z_i^{n_i}Z_{r+j}^{m_{r+j}}$, $1\leqslant i\leqslant r$, $1\leqslant j\leqslant s$ and $m_i,n_i,m_{r+j}$ run over all possible integers in $\ZZ_D$, so the number of terms in \eqref{eqn30} is 
\begin{equation}\label{eqn31}
D^r\times D^r\times D^s=D^{2r+s}\leqslant D^{2k}.
\end{equation}
Compactly we write
\begin{equation}\label{eqn32}
\sum_{U} U\GC_0(\IC) U^\dagger=I.
\end{equation}
We call the collection of all $D^{2r}D^{s}$ such unitary operators the \emph{twirling group} (it is easy to see that the collection of such operators form a group), and denote it by $\TC_0(\IC)$. Note that the structure of the twirling group is easy to read from the structure of $\GC_0(\IC)$, namely, if $\GC_0(\IC)=\langle X_1, Z_1,\ldots, X_r,Z_r, Z_{r+1},\ldots, Z_{r+s}\rangle$, then the twirling group $\TC_0(\IC)$ is generated by
\begin{equation}\label{eqn33}
\TC_0(\IC)=\langle X_1, Z_1,\ldots, X_r,Z_r, X_{r+1},\ldots, X_{r+s}\rangle.
\end{equation}



We can show that $r+s=k$, which implies that the length $l=2r+s$ of the classical encryption key is in general smaller than $2k$, but is bounded below by $k$, the lower bound being achieved whenever $r=0$. This is the case whenever the information group $\GC(\IC)$ is Abelian, which means that before the twirling the intermediate subsets were able to extract only classical information about the secret. Our result is summarized by the following Lemma.
\begin{lemma}\label{lma1}
Let $\GC(\IC)$ be the intermediate information group obtained from an $[[n,k,\delta]]_D$ stabilizer quantum error correcting code. Let \mbox{$\GC_0(\IC)=\langle X_1,Z_1,\ldots,X_r,Z_r,Z_{r+1},\ldots,Z_{r+s}\rangle$} be the canonical intermediate information group isomorphic to $\GC(\IC)$. Then
\begin{equation}\label{eqn34}
r+s=k.
\end{equation} 
\end{lemma}
\begin{proof}
We prove the Lemma by contradiction. Assume \mbox{$r+s<k$}. Split the input qudits into 2 subsets, $I_1$ and $I_2$, with $I_1$ consisting of the first $r+s$ qudits and with $I_2$ consisting of the last $k-(r+s)$ qudits. Choose $R$ to be some subset that belongs to the ramp structure.

We first show that the complement $\bar R$ of $R$ must also belong to the ramp structure. 
There are 3 possible cases for $\bar R$ to belong to : i) the un-authorized structure; ii) the access structure;  iii) the ramp structure;. Case i) must be excluded, since it implies (see the ``Somewhere Theorem" of \cite{PhysRevA.76.062320}) that the complement $R$ of $\bar R$ must belong to the access structure, which contradicts the hypothesis that $R$ belongs to the ramp structure. Also case ii) must be excluded, since it implies (see the ``No Splitting Theorem" of \cite{PhysRevA.76.062320})  that the complement $R$ of $\bar R$ must belong to the un-authorized structure, which again contradicts the hypothesis that $R$ belongs to the ramp structure. We therefore conclude that case iii) is the only possible one, i.e. $\bar R$ must belong to the ramp structure.

Let us return now to the collection $I_2$. The structure of $\GC_0(\IC)$ implies that no information about $I_2$ is present in any subset that belongs to the ramp structure. Denote by $R$ such a subset. Then, by the ``Somewhere Theorem" of \cite{PhysRevA.76.062320}, all information about $I_2$ must be present in the complement $\bar R$ of $R$. But we just proved above that the complement of $R$ must belong to the ramp structure, hence $\bar R$ can extract partial information about $I_2$, a contradiction. Hence the initial hypothesis $r+s<k$ must be false. But $r+s$ cannot be greater than $k$, therefore $r+s=k$.
\end{proof}

\section{Examples}\label{sct6}
\subsection{The CNOT scheme}
The simplest example that illustrates our point is a quantum $[[2,1,1]]_2$ code arising from a CNOT-type isometry, illustrated Fig.~\ref{fgr3}.
\begin{figure}
	\includegraphics[scale=0.7]{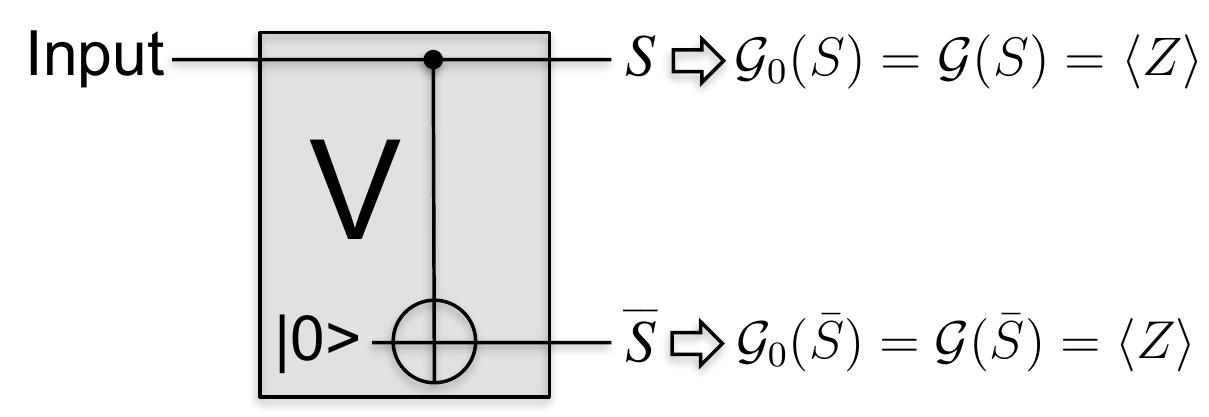}
	\caption{A CNOT-based encoding}
	\label{fgr3}
\end{figure}
Although the code has distance 1 and is not really useful for quantum error correction, it illustrates the basic principles of our work in a very simple and intuitive manner.
The codewords are
\begin{align}\label{eqn35}
\ket{c_0}&=\mathrm{CNOT}(\ket{0}\otimes\ket{0})=\ket{00}\notag\\
\ket{c_1}&=\mathrm{CNOT}(\ket{1}\otimes\ket{0})=\ket{11}.
\end{align}
It is clear that if both $\pt$ and $\ptc$ come together they can reconstruct any secret that was encoded by the dealer (they simply ``undo" the effect of CNOT by applying it again). However, any individual player (consider just $\pt$, since by symmetry the situation is similar for $\ptc$)  has some partial information about the secret. The subset information group in this case is generated by 
\begin{equation}\label{eqn36}
\GC_0(\pt)=\GC_0(\ptc)=\GC(\pt)=\GC(\ptc)=\langle Z \rangle,
\end{equation}
so the player $\pt$ (or $\ptc$) can only extract $Z$-information about the secret. This is easy to verify, since the eigenvectors of the $Z$ operator, $\ket{0}$ and $\ket{1}$, are encoded into $\ket{00}$ and $\ket{11}$, respectively, and the individual players can distinguish with certainty whether $\ket{0}$ or $\ket{1}$ was fed in at the input. On the other hand, the $X$-type of information is totally absent from both $\pt$ and $\ptc$, and this can easily be seen by noting that $\ket{+}=(\ket{0}+\ket{1})/\sqrt{2}$ and $\ket{-}=(\ket{0}-\ket{1})/\sqrt{2}$ are encoded into 2 Bell states which are locally indistinguishable. The $Y$-type of information is also locally absent, by the same argument. Therefore we have
\begin{align}\label{eqn37}
\AC&=\left\{\{1,2\}\right\},\\
\FC&=\{\emptyset\},\notag\\
\IC&=\{\{1\},\{2\}\}\notag.
\end{align}

To hide the partial information from the ramp structure, the dealer randomly generates an integer $m$ and then applies 
the operator $X^m$ before encoding (note that it is essential that $X$ does not commute with the generator $Z$ of the information group ); the twirling group is generated by $\langle X\rangle$.  The dealer then distributes the bit $m$ to the 2 players using a classical $(2,2)$ threshold secret sharing scheme, see Fig.~\ref{fgr4} for a graphical description. 

\begin{figure}
	\includegraphics[scale=0.7]{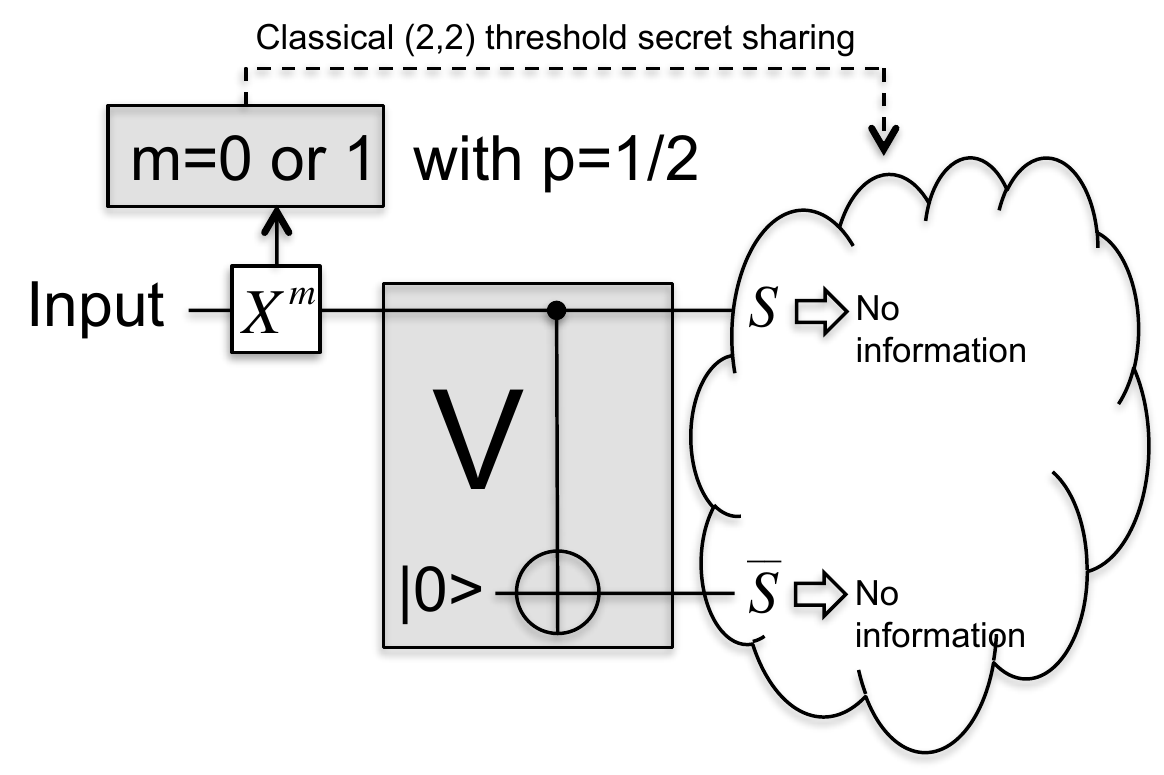}
	\caption{``Cutting" the ramp structure by twirling.}
	\label{fgr4}
\end{figure}

No individual player can recover $m$ (the classical key), hence cannot recover any information about the secret. If the two players collaborate they can then recover $m$, undo the effect of the twirling, then fully recover the quantum secret. Therefore this code is turned into an optimal $(2,2)$ semi-quantum threshold secret sharing scheme.  Note that the number of classical bits required is half compared to the scheme in which the input is fully depolarized by a twirling group $\langle X, Z\rangle$. The length of the classical key is actually achieving the lower bound $l=k=1$.

This example is extremely simple but illustrates our main concepts, and the interested reader can easily work out the details.

\subsection{The $n$-partite GHZ scheme}
We consider now a generalization of the CNOT encoding. The underlying structure is a stabilizer code $[[n,1,1]]_2$ with 2 codewords,
\begin{align}\label{eqn38}
\ket{c_0}&=\ket{00\cdots 0}\notag\\
\ket{c_1}&=\ket{11\cdots 1},
\end{align}
and the encoding circuit can be realized as a ``cascade" of CNOT gates, see Fig.~\ref{fgr5}. 
\begin{figure}
	\includegraphics[scale=0.7]{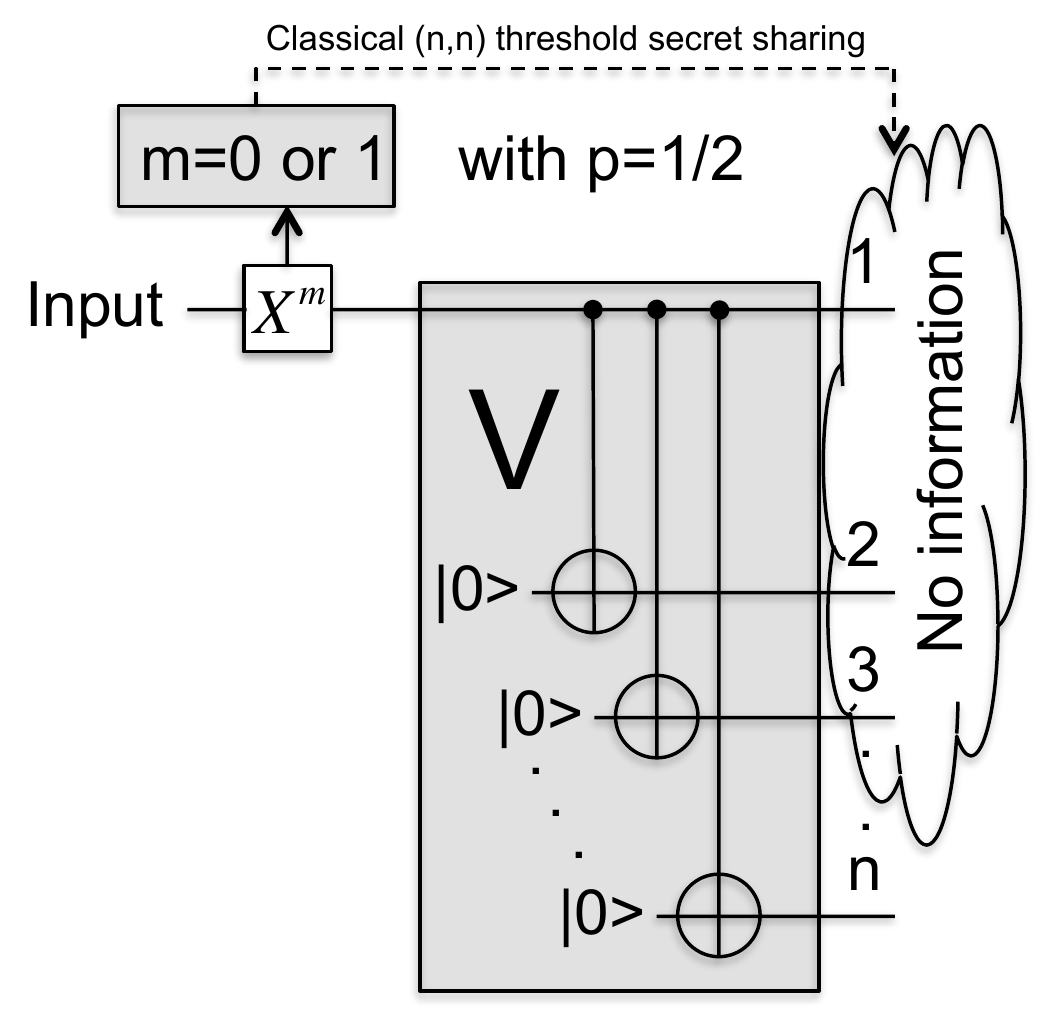}
	\caption{GHZ encoding. ``Cutting" the ramp structure by twirling.}
	\label{fgr5}
\end{figure}
This code has still distance 1, but nevertheless can correct for bit-flip errors using a majority-voting based decoding. The entire collection of $n$ players can fully reconstruct the secret, whereas any subset of less than $n$ players belongs to the ramp structure and can only recover $Z$-information about the secret. More technically,
\begin{equation}\label{eqn39}
\GC_0(\pt)=\GC(\pt)=\langle Z \rangle,\quad\forall \pt \text{ with } |\pt|<n,
\end{equation}
hence
\begin{align}\label{eqn40}
\AC&=\left\{\{1,\ldots,n\}\right\},\\
\FC&=\{\emptyset\},\notag\\
\IC&=\left\{ S\in\NC : |S|<n \right\}\notag.
\end{align} 
The ramp structure can be ``cut" by a twirling with $X^m$ on the input  (so the twirling group is generated again by $\langle X\rangle$), followed by the encoding of the bit $m$ into an $(n,n)$ classical threshold secret sharing scheme. Therefore we end up with an $(n,n)$ semi-quantum threshold secret sharing scheme , with a classical encoding key of length $l=k=1$ (the classical communication required is minimal and equal to the number $k$ of input qubits).

\section{Conclusion and open questions}\label{sct7}
We showed that an $[[n,k,\delta]]_D$ qudit stabilizer code induces a generalized secret sharing scheme consisting of 3 structures: an access structure $\AC$ of which subsets of players can fully recover the secret, a forbidden structure $\FC$ of which subsets cannot recover any information about the secret, and an intermediate structure $\IC$ of which subsets can only recover partial information about the secret. Using an approach based on the theory of qudit stabilizer codes we provided a systematic way of determining the collections $\AC$, $\FC$ and $\IC$. We proved that the information available to a subset (or a collection of subsets) of players can be fully described by an information group, a subgroup of the Pauli group of $k$ qudits, and this quantifies the amount of accessible information in the subset. The structure of the information group provides a natural way of ``scrambling" (or ``twirling") the quantum secret: the dealer applies a unitary operator randomly chosen from the twirling group, completely determined by the information group. The twirling group is generated by $k\leqslant 2r+s\leqslant 2k$ generators, hence the twirling operators are indexed by $2r+s$ integers in $\ZZ_D$, for a total number of $D^{2r+s}$. The dealer can conceal the information from the ramp structure by distributing the $2r+s$-length key (that specifies which twirling operator he applied) to the set of $n$ players using an appropriate perfect classical secret sharing scheme. In particular, we showed that any $[[n,k,\delta]]_D$ stabilizer code induces a semi-quantum $(n-\delta+1,n)$ threshold secret sharing scheme. Our scheme is optimal in terms of the length of the classical encoding key the dealer has to distribute to the players, in contrast to the obvious scheme of $2k$-length key, in which the twirling group is the full Pauli group of $k$ qudits. Our method allows in general for better perfect classical secret sharing scheme encodings of the classical key, and therefore may drastically reduce the total amount of classical communication.

Our scheme is extremely flexible and allows for the construction of more general secret sharing schemes, not just perfect ones. For example, suppose we are interested in hiding partial information only from a collection $\JC\subset\IC$ of subsets of players, not necessarily the entire ramp structure. Then it is enough to find the group $\GC(\IC)$, which is constructed as the union of all subset information groups that correspond to the subsets in the collection, then apply the same algorithm as  before, but now to $\GC(\JC)$ instead of $\GC(\IC)$. We can therefore ``cut" the information about the secret from any collection of subset of players we are interested in. 

Our formalism can also be applied in entanglement sharing schemes, in which the goal of the dealer is to distribute entanglement to subsets of players in such a way that any given subset is either fully entangled with the dealer or otherwise \texttt{¥}heir joint state is separable across the dealer/players cut. This is equivalent to the fact that for any subset, the corresponding information group must be either Abelian or the entire Pauli group $\PC_k$. If this is not the case, we can again employ the notion of twirling and classical secret sharing to transform the intermediate subset group to an Abelian one.

A central issue we did not address in the current article is the recovery operation. In principle, since an authorized set has full information about the secret, a recovery channel always exists, but its construction may not be obvious. In this article we adopt the common strategy used in the search for good quantum error correcting codes \cite{PhysRevA.78.042303}, in which one is not interested in the decoding but only in the parameters of the code. It would be nice to find a clean and systematic way of explicitly constructing this recovery operation. 

In the present article we made heavy use of the stabilizer structure of the encoding isometry. It would be interesting to move beyond stabilizer encodings, or use the formalism of approximate access structures \cite{PhysRevA.78.032330}, i.e. a subset is authorized if it can recover the secret with some bounded error. For the latter problem one should definitely use more general encoding isometries, since the stabilizer ones induce quantum channels with integer capacities, and this is the subject of future work. 

Finally it is interesting to note that the structure of the intermediate information group $\GC(\IC)$ is similar to that of non-Abelian groups used in entanglement-assisted quantum error correcting codes \cite{ToddBrun10202006}, and investigating the relations between the former and the latter may prove fruitful.

\begin{acknowledgments} 
I thank Gilad Gour and Barry Sanders for useful discussions and Robert Spekkens for suggesting the extension of the current scheme to approximate quantum secret sharing. I also thank Lvzhou Li for pointing out some typos in the manuscript. The research described here was supported by the Natural Sciences and Engineering Research Council (NSERC) of Canada and by a Pacific Institute for Mathematical Sciences (PIMS) Postdoctoral Fellowship. 
\end{acknowledgments}


%

\end{document}